\documentclass[11pt]{article}
\usepackage{amssymb}
\usepackage[perpage,symbol]{footmisc}
\usepackage{listings}
\usepackage{amsfonts}
\usepackage{enumerate}
\usepackage{amsmath}
\usepackage{cite}
\usepackage{array}
\usepackage{booktabs}

\begin{document}

\newtheorem{theorem}{Theorem}[section]
\newtheorem{corollary}[theorem]{Corollary}
\newtheorem{definition}[theorem]{Definition}
\newtheorem{proposition}[theorem]{Proposition}
\newtheorem{lemma}[theorem]{Lemma}
\newtheorem{example}[theorem]{Example}
\newtheorem{conjecture}[theorem]{Conjecture}
\newenvironment{proof}{\noindent {\bf Proof.}}{\rule{3mm}{3mm}\par\medskip}
\newcommand{\remark}{\medskip\par\noindent {\bf Remark.~~}}
\title{Linearized Reed-Solomon codes and Linearized Wenger graphs}
\author{Haode Yan\footnote{Dept. of Math., Shanghai Jiao Tong Univ., Shanghai 200240, hdyan@sjtu.edu.cn.},\  Chunlei Liu\footnote{Corresponding author, Dept. of Math., SJTU, Shanghai 200240, 714232747@qq.com.}}
\date{}
\maketitle
\thispagestyle{empty}

\abstract{Let $m$, $d$ and $k$ be positive integers such that $k\leq \frac{m}{e}$, where $e=(m,d)$. Let $p$ be an prime number and $\pi$  a primitive element of ${\mathbb F}_{p^m}$. To each $\vec {a}=(a_{0}, \cdots, a_{k-1})\in \mathbb{F}_{p^{m}}^{k}$, we associate the linearized polynomial
$$f_{\vec {a}}(x)=\sum_{j=0}^{k-1}a_{j}x^{p^{jd}}.$$
And, to each $f_{\vec a}(x)$, we associated the sequence
\[c_{\vec {a}}=(f_{\vec {a}}(1),f_{\vec {a}}(\pi), \cdots , f_{\vec {a}}(\pi^{p^m-2})). \]
Let \[C=\{c_{\vec a}\mid~\vec{a}\in{\mathbb F}_{p^m}^k\}\] be the cyclic code formed by the sequences $c_{\vec {a}}$'s. We call the dual code of $C$ a linearized Reed-Solomon code. The weight distribution of the code $C$
is determined in the present paper.

 Associated to the $k$-tuple $g=(xy,x^{p^d}y,\cdots,x^{p^{(k-1)d}}y)$ of polynomials in ${\mathbb F}_{p^m}[x,y]$, there is a Wenger graph $W_{p^m}(g)$. The spectrum of the graph $W_{p^m}(g)$ is also determined in the present paper.



\section{\small{INTRODUCTION}}
Let $m$, $d$ and $k$ be positive integers such that $k\leq \frac{m}{e}$, where $e=(m,d)$. Let $p$ be an prime number and $\pi$  a primitive element of ${\mathbb F}_{p^m}$. To each $\vec {a}=(a_{0}, \cdots, a_{k-1})\in \mathbb{F}_{p^{m}}^{k}$, we associate the linearized polynomial
$$f_{\vec {a}}(x)=\sum_{j=0}^{k-1}a_{j}x^{p^{jd}}.$$
And, to each $f_{\vec a}(x)$, we associated the sequence
\[c_{\vec {a}}=(f_{\vec {a}}(1),f_{\vec {a}}(\pi), \cdots , f_{\vec {a}}(\pi^{p^m-2})). \]
Let \[C=\{c_{\vec a}\mid~\vec{a}\in{\mathbb F}_{p^m}^k\}\] be the cyclic code formed by the sequences $c_{\vec {a}}$'s. We call the dual code of $C$ a linearized Reed-Solomon code. We are interested in the weight distribution of the code $C$. Other kinds of cyclic codes constructed from linearized polynomials were studied in \cite{FL,Tr,ZWHZ,ZDLZ}.

Our preliminary task is to determine the weight set of the code $C$. The result is  the following.
\begin{theorem}\label{valueset} If $\vec {a}\neq \vec{0}$, then
\[w(c_{\vec {a}}) \in \{p^m-p^{er}|\  0\leq r\leq k-1 \},\]
where $w(c_{\vec a})$ is the Hamming weight of $c_{\vec a}$.
In particular, $w(c_{\vec {a}})=0$ if and only if $\vec {a}=0$.
\end{theorem}
Our main task is to determine the frequencies
 \begin{equation}\label{fredef} n_{r}=|\{ \vec {a} \in \mathbb{F}_{p^{m}}^{k} |~\ w(c_{\vec {a}})=p^m-p^{er}\} |,~0\leq r \leq k-1.\end{equation}
The result is the following.
\begin{theorem}\label{frequency}We have, for $0\leq r \leq k-1$,
\[ n_{r}=\binom{\frac{m}{e}}{r}_{p^e}\sum_{i=0}^{k-r-1} (-1)^{i}p^{\frac{ei(i-1)}{2}}\binom{\frac{m}{e}-r}{i}_{p^e}(p^{m(k-r-i)}-1), \]
where $\binom{n}{i}_{q}$ is the number of $i$-dimensional $\mathbb{F}_{q}$-subspaces of $\mathbb{F}_{q}^{n}$.
\end{theorem}
\paragraph{}
 Associated to the $k$-tuple $g=(xy,x^{p^d}y,\cdots,x^{p^{(k-1)d}}y)$ of polynomials in ${\mathbb F}_{p^m}[x,y]$, there is a Wenger graph $W_{p^m}(g)$. The spectrum of that Wenger graph will be determined in the last section.
\paragraph{}In the study of Gaussian binomial coefficients, we find the following conjecture.
\begin{conjecture}We have \[\binom{u}{i}_{q^2}\sum_{j=0}^iq^{j}\binom{i}{j}_{q^2}
=\binom{u}{i}_{q}\prod_{j=0}^{i-1}(1+q^{u-j}).
\]\end{conjecture}

\section{\small{PROOF OF THEOREM \ref{valueset}}}
Denote the set of all the roots of $f_{\vec {a}}$  by
\[{\rm Null}(f_{\vec {a}})=\{x \in  \mathbb{F}_{p^{m}}|\ f_{\vec {a}}(x)=0\}.\]
It is easy to see that
\begin{equation}\label{Eqs:2.1}
w(c_{\vec {a}})=p^m-|{\rm Null}(f_{\vec {a}})|.
\end{equation}
It follows from the above equality that Theorem \ref{valueset} is equivalent to the following.
\begin{lemma}\label{nullset}If $\vec {a}\neq \vec{0}$, then
\[|{\rm Null}(f_{\vec {a}})| \in \{p^{er}| \  0\leq r\leq k-1 \}.\]
\end{lemma}
\begin{proof} Suppose that $\vec {a} \neq 0$. Note that $\{x \in  \mathbb{F}_{p^{md/e}}|f_{\vec {a}}(x)=0\}$ is a subspace of $\mathbb{F}_{p^{md/e}}$ over ${\mathbb F}_{p^d}$ of dimension $\leq (k-1)$.
As $(m, d)=e$,
a basis of $\mathbb{F}_{p^{m}}$ over ${\mathbb{F}_{p^{e}}}$ is also a basis of $\mathbb{F}_{p^{md/e}}$ over $\mathbb{F}_{p^{d}}$.
It follows that
\[
 {\rm dim}_{\mathbb{F}_{p^e}}{\rm Null}({f_{\vec {a}}})\leq (k-1).
 \]
Hence $|{\rm Null}(f_{\vec {a}})|\in \{p^{er}|~ 0\leq r\leq k-1 \}$. Lemma \ref{nullset} is proved.
\end{proof}

It is obvious that (\ref{Eqs:2.1}) also implies the following.
\begin{lemma}\label{corollary}For $0\leq r \leq k-1$, we have
 \[ n_{r}=|\{ \vec {a} \in \mathbb{F}_{p^{m}}^{k}|\ {\rm Null}({f_{\vec {a}}})=p^{er}\} |.\]
\end{lemma}

\section{\small{$q$-BINOMIAL M\"{O}BIUS INVERSION FORMULA}}
Let $q$ be a prime power. In this section, we prove the $q$-binomial M\"{o}bius inversion formula. We begin with the following definition.
\begin{definition}[$q$-binomial M\"{o}bius function]\label{mobiusdef} The $q$-binomial M\"{o}bius  function is a function on the set of finite-dimensional vector spaces over $\mathbb{F}_q$ defined by the formula
\[ \mu_q(U)=(-1)^{\rm{dim} U}q^{\binom{\rm{dim} U}{2}}.\]
\end{definition}
\begin{lemma}\label{mobiusprop} We have
\[\sum_{V\subseteq U}\mu_q(V)=\left\{
                                \begin{array}{ll}
                                  1, & \hbox{} U=\{0\},\\
                                  0, & \hbox{} U\neq\{0\}.
                                \end{array}
                              \right.
\]
\end{lemma}\begin{proof}
The lemma follows from the Gaussian binomial formula
\[\prod_{i=0}^{n-1}(1+q^it)=\sum_{i=0}^nq^{\frac{i(i-1)}{2}}\binom{n}{i}_qt^i,\]
which can be proved by induction on $n$.\end{proof}
\begin{lemma}[$q$-binomial M\"{o}bius inversion formula]\label{mobius} Let $f,g$ be functions defined on the set of $\mathbb{F}_q$-vector spaces. We have
\[ g(U)=\sum_{V\subseteq U}f(V) \]
if and only if
\[ f(U)=\sum_{V\subseteq U}\mu_q(U/V)g(V).\]
\end{lemma}

\begin{proof} If
\[ g(U)=\sum_{V\subseteq U}f(V), \]
then, by Lemma \ref{mobiusprop},
\[\begin{split}\sum_{V\subseteq U}\mu_q(U/V)g(V)
&=\sum_{V\subseteq U}\mu_q(U/V)\sum_{W\subseteq V}f(W)\\
&=\sum_{W\subseteq U}f(W)\sum_{W\subseteq V\subseteq U}\mu_q(U/V)\\
&=\sum_{W\subseteq U}f(W)\sum_{X\subseteq U/W}\mu_q(X)\\
&=f(U).\end{split}\]
The converse can be proved similarly.
\end{proof}

\section{\small{q-LINEARIZED VAN DER MONDE MATRIX}}
In this section, we will determine the rank of  a $q$-linearized Van Der Monde matrix.
\begin{definition}\label{van}If $x_{1}, x_{2},\cdots,x_{k}$ are $\mathbb{F}_{p^{e}}$-linearly independent elements in  $\mathbb{F}_{p^{m}}$, then matrix\[
\left( \begin{array}{cccc}
x_1&{x_1}^{p^d}&\cdots&{x_1}^{p^{(k-1)d}}\\
x_2&{x_2}^{p^d}&\cdots&{x_2}^{p^{(k-1)d}}\\
\vdots&\vdots&\ddots&\vdots\\
x_k&{x_k}^{p^d}&\cdots&{x_k}^{p^{(k-1)d}}\\
\end{array} \right)\] is called q-linearized Van Der Monde matrix.
\end{definition}
The following lemma is a slight generalization of a lemma of Cao-Lu-Wan-Wang-Wang \cite{WAN}.
\begin{lemma}\label{rank} A $q$-linearized Van Der Monde matrix is of full rank.
\end{lemma}
\begin{proof} Let $H$ be the $\mathbb{F}_{p^{e}}$-subspace of ${\mathbb F}_{p^m}$ spanned by the elements $x_{1},x_{2},\cdots, x_{k}$.  By Lemma \ref{nullset},
\[\{ \vec {a} \in \mathbb{F}_{p^{m}}^{k}|\ {\rm Null}(f_{\vec {a}})\supseteq H\}=\{0\}.\] Now we consider the following homogeneous linear system of equations:\[
\begin{cases}
a_{0}x_{1}+a_{1}{x_{1}}^{p^{d}}+\cdots + a_{k-1}{x_{1}}^{p^{(k-1)d}}=0\\
a_{0}x_{2}+a_{1}{x_{2}}^{p^{d}}+\cdots + a_{k-1}{x_{2}}^{p^{(k-1)d}}=0\\
\cdots\\
a_{0}x_{k}+a_{1}{x_{k}}^{p^{d}}+\cdots + a_{k-1}{x_{k}}^{p^{(k-1)d}}=0
\end{cases}\]
Obviously, ${\rm Null}(f_{\vec {a}})\supseteq H $  if  and  only  if $\vec {a}$ is a solution of the above system,
which implies the rank of the coefficient matrix is  $k$.  Lemma \ref{rank} is proved.
\end{proof}
From the above lemma, one immediately deduces the following.
\begin{corollary}\label{rankcor}If $r\leq k$ and $x_{1},x_{2},\cdots,x_{r}$ are $\mathbb{F}_{p^{e}}$-linearly independent elements in  $\mathbb{F}_{p^{m}}$, then the following matrix of full rank\[
\left( \begin{array}{cccc}
x_1&{x_1}^{p^d}&\cdots&{x_1}^{p^{(k-1)d}}\\
x_2&{x_2}^{p^d}&\cdots&{x_2}^{p^{(k-1)d}}\\
\vdots&\vdots&\ddots&\vdots\\
x_r&{x_r}^{p^d}&\cdots&{x_r}^{p^{(k-1)d}}\\
\end{array} \right)\]
\end{corollary}
\section{\small{PROOF OF THEOREM \ref{frequency}}}
In this section, we shall prove Theorem \ref{frequency}. By the theory of Delsarte \cite{Del}, Theorem \ref{frequency} follows from
Theorem \ref{valueset} directly. However we would like to take an alternative approach. We begin with the following notation.
For an ${\mathbb F}_{p^e}$-subspace $V$ of ${\mathbb F}_{p^m}$, we write
\[C_{V}=\{ \vec {a} \in \mathbb{F}_{p^{m}}^{k}|\ {\rm Null}(f_{\vec {a}})\supseteq V^\bot\}.\]
\begin{lemma}\label{cvvalue}
We have
\[
|C_{V}|=
\begin{cases}
1, &\ {k \leq {\rm dim}V^\bot \leq \frac{m}{e}},\\
p^{m(k-{\rm dim}V^\bot)}, &\ {{\rm dim}V^\bot < k}.
\end{cases}
\]
\end{lemma}
\begin{proof}
Let $\{x_{1},x_{2},...,x_{r}\}$ be a basis of $V^\bot$. Then $C_{V}$ is precisely the set of $\vec {a} \in \mathbb{F}_{p^{m}}^{k}$ satisfying the following system
\[
\begin{cases}
a_{0}x_{1}+a_{1}{x_{1}}^{p^{d}}+\cdots + a_{k-1}{x_{1}}^{p^{(k-1)d}}=0\\
a_{0}x_{2}+a_{1}{x_{2}}^{p^{d}}+\cdots + a_{k-1}{x_{2}}^{p^{(k-1)d}}=0\\
\cdots\\
a_{0}x_{r}+a_{1}{x_{r}}^{p^{d}}+\cdots + a_{k-1}{x_{r}}^{p^{(k-1)d}}=0.\\
\end{cases}\] Lemma \ref{cvvalue} now follows from Lemma \ref{rankcor}.
\end{proof}

We now prove Theorem \ref{frequency}.\\
\begin{proof}
For an ${\mathbb F}_{p^e}$-subspace $V$ of ${\mathbb F}_{p^m}$, we write
\[S_{V}=\{ \vec {a} \in \mathbb{F}_{p^{m}}^{k} |~\ {\rm Null}(f_{\vec {a}})=V^{\bot}\}.\]
By definition,
\[|C_V|=\sum_{W\subseteq V}|S_W|.\]
By Lemmas \ref{mobius}, \ref{mobiusprop}, and \ref{cvvalue},
\[\begin{split}|S_{V}|&=\sum_{W \subseteq V}\mu_{p^e}(V/W)|C_{W}|\\
&=\sum_{W \subseteq V}\mu_{p^e}(V/W)(|C_{W}|-1)\\
&=\sum_{\stackrel{W \subseteq V}{\dim W<k}}\mu_{p^e}(V/W)(p^{m(k-{\rm dim}W^\bot)}-1).\end{split}\]
By Lemma \ref{corollary}, we have
\[\begin{split} n_{r}&=\sum_{dimV^\bot=r}|S_{V}|\\
&=\binom{\frac{m}{e}}{r}_{p^e}\sum_{i=0}^{k-r-1} (-1)^{i}p^{\frac{ei(i-1)}{2}}\binom{\frac{m}{e}-r}{i}_{p^e} (p^{m(k-r-i)}-1).\end{split}\]
Theorem \ref{frequency} is proved.\end{proof}
\section{\small{LINEARIZED WENGER GRAPHS}}

Let $m$, $d$ and $k$ be positive integers such that $k\leq \frac{m}{e}$, where $e=(m,d)$. Let $p$ be a prime number. To each $k$-tuple $g=(g_0,g_1,\cdots,g_{k-1})$ of polynomials in $\mathbb{F}_{p^{m}}[x,y]$, we associate a graph $W_{p^m}(g)$, which is called the Wenger graph, as follows. Let $P$ and $L$ be two copies of $\mathbb{F}_{p^{m}}^{k+1}$.  The vertex set $V$ of $W_m(g)$ is $P\cup L$. The edge set $E$ of $W_{p^m}(g)$ consists of $(p,l)\in P\times L$ satisfying
\[l_1+p_1=g_0(p_0,l_0),~l_2+p_2=g_1(p_0,l_0),\cdots,~l_{k}+p_{k}=g_{k-1}(p_0,l_0),\]
where $p=(p_0,p_1,\cdots ,p_{k})$ and $l=(l_0,l_1,\cdots ,l_{k})$.
\paragraph{}
The above definition of Wenger graph was given by Viglione \cite{V1}.
It is a generalization of the original Wenger graph. The original Wenger graph was introduced by Wenger \cite{W}, and was studied by many authors such as Lazebnik-Ustimenko \cite{LU,LU2}, Viglione \cite{V1,V2}, and Cioab\u{a}-Lazebnik-Li\cite{CLL}.

\paragraph{}
The Wenger graph we are interested in is the Wenger graph associated to the $k$-tuple $g=(xy,x^{p^d}y,\cdots,x^{p^{(k-1)d}}y)$. Following Cao-Lu-Wan-Wang-Wang \cite{WAN}, we call that Wenger graph a linearized Wenger graph.
We need the following lemma.
\begin{lemma}[Cao-Lu-Wan-Wang-Wang \cite{WAN}]\label{wan} Let $W_{p^m}(g)$ be the Wenger graph associated to the $k$-tuple $(f_0(x)y,f_1(x)y,\cdots,f_{k-1}(x)y)$, where $(f_0, f_1,\cdots , f_{k-1})$ is a $k$-tuple of polynomials in ${\mathbb F}_{p^m}[x]$ such that the map \[{\mathbb F}_{p^m}\rightarrow{\mathbb F}_{p^m}^k,~u \mapsto (f_0(u), f_1(u),\cdots , f_{k-1}(u))\] is injective. Then, counting multiplicities, the eigenvalues of  $W_{p^{m}}(g)$ are
\[\{\pm\sqrt{p^m N_{F_{ \tilde{a}}}}\mid~   \tilde{a}=( a_{-1}, a_0, \cdots , a_{k-1} ) \in \mathbb{F}_{p^{m}}^{k+1}\},\]
where \[F_{ \tilde{a}}(x)= a_{-1}+ a_0 f_0(x)+a_1 f_1(x) +\cdots + a_{k-1} f_{k-1}(x),\] and \[N_{F_{ \tilde{a}}}=|\{x \in \mathbb{F}_{p^{m}} |\  F_{ \tilde{a}}(x)=0 \}|.\]\end{lemma}

The following theorem generalizes a result of  Cao-Lu-Wan-Wang-Wang \cite{WAN}. It also solves the open problem put forward by  Cao-Lu-Wan-Wang-Wang \cite{WAN}.
\begin{theorem}\label{graph} Let $W_{p^m}(g)$ the Wenger graph associated to the $k$-tuple \[g=(xy,x^{p^d}y,\cdots,x^{p^{(k-1)d}}y).\] Then the  eigenvalues of $W_{p^{m}}(g)$ are
\[0,\pm p^m,  \pm\sqrt{p^{m+er}}, 0\leq r \leq k-1. \]
Moreover, the multiplicities of the eigenvalues $\pm p^m$ are $1$, the multiplicities of the eigenvalues $\pm \sqrt{p^{m+er}}$ are $p^{m-er}n_r$, and the multiplicity of the eigenvalue 0 is \[\sum_{r=1}^{k-1} (p^m-p^{m-er})n_r,\]
where the $n_r$'s are the frequencies defined in (\ref{fredef}).
\end{theorem}
\begin{proof} It is obvious that
\[N_{F_{ \tilde{a}}}=\left\{
                       \begin{array}{ll}
                         |{\rm Null}(f_{\vec {a}})|, & \hbox{}-a_{-1}\in {\rm Image}(f_{\vec a}),\\
                         0, & \hbox{}-a_{-1}\notin {\rm Image}(f_{\vec a}).
                       \end{array}
                     \right.\]
By Lemma \ref{nullset}, the distinct eigenvalues of $W_{p^m}(g)$ are
\[0,\pm p^m,  \pm\sqrt{p^{m+er}}, 0\leq r \leq k-1. \]
It remains to calculate the multiplicity of each eigenvalue. First, it is easy to see that the multiplicities of the eigenvalues $\pm p^m$ are $1$. Secondly, for each $0\leq r \leq k-1$, the multiplicities of the eigenvalues  $\pm\sqrt{p^{m+er}}$ are the number of $\tilde{a}$ such that
$|{\rm Null}(f_{\vec {a}})|=p^{er}$ and $-a_{-1}\in {\rm Image}(f_{\vec a})$, which is equal to $p^{m-er}n_r$ by Lemma \ref{corollary}.
Finally,  the multiplicity of the eigenvalue $0$ is equal to
\[\sum_{r=1}^{k-1} (p^m-p^{m-er})n_r.\]
Theorem \ref{graph} is proved.
\end{proof}

\paragraph{}
{\bf Acknowledgements.} The authors thanks Daqing Wan for mentioning to us the relationship between linearized polynomials and Wenger graphs.

\end{document}